%% file: axioms.tex
\documentclass{paper}
\usepackage[margin=1in]{geometry}
\pdfoutput=1
\usepackage{mathpazo}
\usepackage{amsmath,amsthm,amssymb}
\usepackage{xspace}
\usepackage{microtype} 
\usepackage{xcolor}
\usepackage{booktabs}

\newtheorem{definition}{Definition}[section]

\newtheorem{theorem}{Theorem}[section]

\newtheorem{axiom}{Axiom}[section]

\usepackage{url}
\usepackage{mathtools}
\mathtoolsset{showonlyrefs,mathic}
\usepackage{cite}
\usepackage{xcolor}
\usepackage{hyperref}
\usepackage{centernot}
\usepackage[normalem]{ulem}
\usepackage{enumitem}
\setlist[enumerate,1]{%
  label=\arabic*.,
}
\newlist{inlinelist}{enumerate*}{1}
\setlist*[inlinelist,1]{%
  label=(\roman*),
}

\usepackage[final,inline,nomargin,index]{fixme}
\fxsetup{theme=colorsig,mode=multiuser,inlineface=\itshape,envface=\itshape}
\FXRegisterAuthor{sv}{asv}{Suresh}
\FXRegisterAuthor{sf}{asf}{Sorelle}
\FXRegisterAuthor{cs}{acs}{Carlos}

\usepackage{graphicx,xspace}

\usepackage{acro}
\DeclareAcronym{ts}{
short  = \text{CS},
long = construct space,
short-format = \scshape
}
\newcommand{\truespace}{\ensuremath{\mathcal{CS}}\xspace}

\DeclareAcronym{os}{
  short  = \text{OS},
  long = observed space,
  short-format = \scshape
}
\newcommand{\obsspace}{\ensuremath{\mathcal{OS}}\xspace}

\DeclareAcronym{wd}{
short  = \text{WD},
long = Wasserstein distance,
short-format = \scshape
}

\DeclareAcronym{gw}{
short  = \text{GWD},
long = Gromov-Wasserstein distance,
short-format = \scshape
}

\DeclareAcronym{ds}{
short  = \text{DS},
long = decision space,
short-format = \scshape
}
\DeclareAcronym{rich}{
short = rich,
long = richness
}

\DeclareAcronym{wysiwyg}{
short = wysiwyg,
long = what you see is what you get,
short-format = \scshape,
long-format = \itshape
}
\DeclareAcronym{wae}{
short = wae,
long = we're all equal,
short-format = \scshape,
long-format = \itshape
}
\DeclareAcronym{ifm}{
short  = IFM,
long = individual fairness mechanism,
short-format = \scshape
}
\DeclareAcronym{gfm}{
short  = GFM,
long = group fairness mechanism,
short-format = \scshape
}

\newcommand{\decisionspace}{\ensuremath{\mathcal{DS}}\xspace}

\newcommand{\wae}{\textsf{WAE}\xspace}
\newcommand{\ip}{\text{\textsf{WM}\xspace}}

\newcommand{\task}{\ensuremath{\mathsf{T}}\xspace}

\newcommand{\eps}{\varepsilon}
\newcommand{\emd}{\ensuremath{\mathcal{W}}}
\newcommand{\gw}{\ensuremath{\mathcal{GW}}}
\newcommand{\gs}{\ensuremath{\sigma}}

\newcommand{\bignote}[1]{%
  \begin{center}
    \noindent\colorbox{gray!20}{%
      \begin{minipage}{\columnwidth}%
        \noindent #1%
      \end{minipage}%
    }
  \end{center}
}

\title{On the (im)possibility of fairness\thanks{This research was funded in part by the NSF under grants IIS-1251049, CNS-1302688, IIS-1513651, IIS-1633724, and IIS-1633387.}}
\author{Sorelle A. Friedler\\ {\large Haverford College}\thanks{\url{sorelle@cs.haverford.edu}} \and Carlos Scheidegger\\ {\large University of Arizona}\thanks{\url{cscheid@cscheid.net}} \and Suresh Venkatasubramanian\\{\large University of Utah}\thanks{\url{suresh@cs.utah.edu}}}
\date{} 

\begin{document}

\maketitle
\begin{abstract}
  What does it mean for an algorithm to be fair?
  Different papers use different notions of algorithmic fairness, and although these appear internally consistent, they also seem mutually incompatible.
  We present a mathematical setting in which the distinctions in previous papers can be made formal.
  In addition to characterizing the spaces of inputs (the ``observed'' space) and outputs (the ``decision'' space), we introduce the notion of a \emph{construct space}: a space that captures unobservable, but meaningful variables for the prediction.
  We show that in order to prove desirable properties of the entire decision-making process, different mechanisms for fairness require different assumptions about the nature of the mapping from construct space to decision space.
  The results in this paper imply that future treatments of algorithmic fairness should more explicitly state assumptions about the relationship between constructs and observations.
\end{abstract}

\input{intro}
\input{spaces}
\input{defns}
\input{impos}
\input{related}

\section{Conclusions}
\label{sec:conclusions}

In this paper, we have shown that some notions of fairness are
fundamentally incompatible with each other.
These results might appear discouraging if one hoped for a universal notion of fairness, but we believe they are important.
They force a shift in the focus of the discussion surrounding algorithmic fairness:
without precise definitions of beliefs about the state of the world
and the kinds of harms one wishes to prevent, our results show that it is not possible to
make progress. They also force future discussions of
algorithmic fairness to directly consider the values inherent in assumptions about how the
observed space was constructed, and that such value assumptions
should always be made explicit.

Although the specific theorems themselves matter, it is the
definitions and the problem setup that are the fundamental
contributions of this paper. This work represents a first step towards
fairness researchers using a shared setting, vocabulary, and
assumptions.

\section{Acknowledgements}
We want to thank the attendees at the Dagstuhl workshop on Data, Responsibly for their helpful comments on an early presentation of this work;  special thanks to Cong Yu and Michael Hay for encouraging us to articulate the subtle differences in reasons for choosing a specific worldview and to Nicholas Diakopoulos and Solon Barocas to pointing us to the relevant work on ``constructs'' and inspiring our naming of that space.  Thanks to Tionney Nix and Tosin Alliyu for generative early conversations about this work.  Thanks also to danah boyd and the community at the Data \& Society Research Institute for continuing discussions about the meanings of fairness and non-discrimination in society.

\bibliographystyle{abbrv}
\bibliography{bias}

\end{document}

%% file: intro.tex
\section{Introduction}
\label{sec:introduction}

Machine learning has embedded itself deep inside many of the decision-making systems that
used to be driven by humans. Whether it's resume filtering for jobs, or
admissions to college, credit ratings for loans, or all components of the
criminal justice pipeline, automated tools are being used to find patterns, make
predictions, and assist in decisions that have significant impact on our lives. 

The ``rise of the machines'' has raised concerns about the fairness of these
processes. Indeed, while one of the rationales for introducing automated
decision making was to replace subjective human decisions by ``objective''
algorithmic methods, a large body of research has shown that machine learning is
not free from the kinds of discriminatory behavior humans display. This area of
\emph{algorithmic fairness} now contains many ideas about how to prevent
algorithms from learning bias\footnote{Here, and in the rest of this paper, we
  will typically use ``bias" to denote discriminatory behavior in society, rather
  than the statistical notion of bias.  Similarly, ``discriminate" will refer to the societal notion.} and how to design algorithms that are
``fairness-aware.''

Strangely though, the basic question ``what does it mean for an algorithm to be
fair?'' has gone under-examined.  While many papers have proposed quantitative measures of fairness,
these measures rest on unstated assumptions about fairness in society. As we
shall show, these assumptions, if brought into the open, are often mutually
incompatible, rendering it difficult to compare proposals for fair algorithms with each other.

\subsection{Our Work}
\label{sec:our-work}

Definitions of fairness, nondiscrimination, and justice in society have been debated in the social science community extensively, from Rawls to Nozick to Roemer, and many others.  A parallel debate is ongoing within the computer science community, including discussions of individual fairness \cite{Dwork12Fairness} vs. group fairness (e.g., disparate impact's four-fifths rule \cite{Feldman2015DisparateImpact} and a difference formulation of a discrimination score \cite{Calders10NaiveBayes, Kamiran09Classifying, Kamishima11Fairness, Ruggieri2014TCloseness, icml2013_zemel13}).  These discussions reveal differences in the understood meaning of ``fairness" in decision-making centering around two different interpretations of the extent to which factors outside of an individual's control should be factored into decisions made about them and the extent to which abilities are innate and measurable.  This tension manifests itself in the debates between ``equality of outcomes" and ``equality of treatment" that have long appeared under many names.  Our contribution to this debate will be to make these definitions mathematically precise and reveal the axiomatic differences at the heart of this debate.  (We will review the literature in light of these definitions in Section \ref{sec:related}.)

In order to make fairness mathematically precise, we tease out the
difference between \emph{beliefs} and \emph{mechanisms} to make clear what
aspects of this debate are opinions and which choices and policies logically
follow from those beliefs.
Our goal is to make the embedded value systems
transparent so that the belief system can be chosen, allowing mechanisms to
be proven compatible with these beliefs.

We will create this separation of concerns by developing a mathematical theory of
fairness in terms of transformations between different kinds of spaces. Our primary insight can be summarized as:
\begin{quote}\emph{
  To study algorithmic fairness is to study the interactions between different
  \emph{spaces} that make up the decision pipeline for a task. 
}\end{quote}
We will first argue that there are more spaces that are implicitly involved in the
decision-making pipeline than are typically specified. In fact, it is the
\emph{conflation} of these spaces that leads to much of the confusion
and disagreement in the literature on algorithmic fairness.
Next, we will reinterpret notions of fairness, structural bias, and non-discrimination as quantifying the way that spaces are transformed to each other.
With this framework in place, we can make formal tensions between fairness and non-discrimination by revealing fundamental differences in worldviews underlying these definitions.

Our specific contributions are as follows:
\begin{itemize}
\item We introduce the idea of (task-dependent) spaces that interact in any
  learning task, specifically introducing the \acl{ts} that captures the notion that features of interest for decision-making are necessarily imperfect proxies for the construct of interest.
\item We reinterpret notions of fairness, structural bias, and non-discrimination mathematically as functions of transformations between these spaces.
\item Surprisingly, we show that fairness can be guaranteed only with very
  strong assumptions about the world: namely, that ``what you see is what you
  get,'' i.e., that we can correctly measure individual fitness for a task regardless
  of issues of bias and discrimination. We complement this with an \emph{impossibility} result, saying that if
  this strong assumption is dropped, then fairness can no longer be guaranteed. 
\item We develop a theory of non-discrimination based on a quantification of structural bias.  Building non-discriminatory decision algorithms is shown to require a different worldview, namely that ``we're all equal," i.e., that all groups are assumed to have similar abilities with respect to the task in the construct space.  
  \item We show that virtually all methods that propose to address algorithmic
  fairness make implicit assumptions about the nature of these spaces and how
  they interact with each other. 
\end{itemize}


%% file: spaces.tex
\section{Spaces: Construct, Observed and Decision}
\label{sec:spac-indiv-decis}

If we consider our guiding informal 
understanding of
fairness - that similar people should be treated similarly - in the context of
algorithm design, we must begin by determining how people will be represented as
inputs to the algorithm, and what associated notion of similarity on this representation is appropriate.  These two choices will entirely determine what we mean by fairness, and there are many subtle choices that must be made in this determination as we build up to a formal definition.  Fairness-aware algorithms (and indeed all algorithms in machine learning) can be viewed as mappings between spaces, and we will adopt this viewpoint. They take inputs from some \emph{feature space} and return outputs in a \emph{decision space}.  The question then becomes how we should define points and the associated metric to precisely define these spaces.  To illuminate some of the subtleties inherent in these choices, we'll introduce what will be a running example of fairness in a college admissions decision.

\bignote{\textbf{Example: Setting up a College Admissions Decision}\\
We can think of the college admissions process as being a procedure that takes a set of people and applies a yes or no decision to each person.  Before determining a fair procedure, an admissions office would need to determine what aspects of a person they wanted to make the decision based on.  This might include potential, intelligence, and diligence.  These choices would lead to other questions:  At what point in a person's life should these aspects of their personality and ability be measured?  Is intelligence set at birth or adaptable and when does this mean it should be measured for a college admissions decision?  How can potential and diligence be accurately measured?
}

In order to define a feature space, we must answer questions about what features should be included and how (and when) they should be measured.  This description illuminates our first important distinction from a common set-up of such a problem: the feature space itself is a representation of a chosen set of possibly hidden or unmeasurable constructs.  Determining which features should be considered is part of the determination of how the decision should be made; representing those constructs in measurable form is a separate and important step in the process.  This distinction motivates our first two definitions.

\begin{definition}[\Acl{ts} (\acs{ts})]
  The \acl{ts} is a metric space $\truespace = (P, d_P)$ consisting of individuals and a distance between them. It is assumed that the distance $d_P$ correctly captures closeness with respect to the task. 
\end{definition}

The \acl{ts} is the space containing the features that we would like to make a decision based on.  These are the ``desired" or ``true" features at the time chosen for the decision, and the ability to accurately measure similarity between people with respect to the task.  This is the space representing the input to a decision-making process, if we had full access to it.

\bignote{\textbf{Example: \Acl{ts} for a College Admissions Decision}\\
Suppose that a college admissions team decided to make admissions decisions based on predicted potential of applicants.  Personal qualities, such as self-control, growth mind-set, and grit, are known to be determining
factors in later success \cite{almlund2011Personality}.  \emph{Grit} is roughly defined as an individual's
ability to demonstrate passion, perseverance and resilience towards their chosen
goal.  In this way, a college admissions decision could attempt to use grit as a feature of ``potential" in the \acl{ts}. 
}

In reality, we might not know these features or even the true similarity between individuals. All we have is what we measure or \emph{observe}, and this leads us to our next definition. 

\begin{definition}[\Ac*{os}]
  The \acl{os} (with respect to \task) is a metric space $\obsspace = (\hat{P}, \hat{d})$. We assume an \emph{observation} process $g : P \to \hat{P}$ that generates an entity $\hat{p} = g(p)$ from a person $p \in \truespace$. 
\end{definition}

\bignote{\textbf{Example: \Acl{os} for a College Admissions Decision}\\
Grit (and other such qualities) are not directly observable: research attempts to measure grit indirectly through proxies such as self-reported surveys \cite{duckworth2007Grit}.  The ability to measure grit and these other qualities precisely is limited \cite{duckworth2015Measurement}. In this setting, the ``amount of grit'' is a feature in
the \acl{ts}: it is hidden from us but appears to have some unknown influence on
the desired predicted outcome. We attempt to infer it through imperfect proxy features, such as a ``survey-based grit score," that lie in the \acl{os}.
}

The final part of a task is a \emph{decision space} of outcomes. 

\begin{definition}[\Ac*{ds}]
  A \acl{ds} is a metric space $\decisionspace = (O, d_O)$, where $O$ is a space of outcomes and $d_O$ is a metric defined on $O$. A task \task can be viewed as the process of finding a map from $P$ or $\hat{P}$ to $O$. 
\end{definition}

\bignote{\textbf{Example: \Acl{ds} for a College Admissions Decision}\\
The \acl{ds} for a college admissions decision consists of the (potentially unobservable, hopefully predictable) information that makes up the final admissions decision.  The decision space might be simply the resulting yes/no (admit / don't admit) decisions, or might be the predicted potential of an applicant or their predicted performance in college (a threshold could then be applied to this \acl{ds} to generate yes/no decisions).
}

\paragraph{How the spaces interact.}
\label{sec:inter-betw-spac}

Algorithmic decision-making is a set of mappings between the three spaces
defined above. The desired outcome is a mapping from \truespace to
\decisionspace via an unknown and complex function $o = f(X_1, X_2, \ldots)$ of features
that lie in the \acl{ts}.

In order to implement an algorithm that predicts the desired outcome, we must
first extract usable data from \truespace: this is a collection of  mappings from \truespace to
\obsspace. The features $Y_1, Y_2, \ldots, Y_\ell$ in \obsspace  might be: 
\begin{itemize}
\item noisy variants of the $X_i: Y_i = g(X_i)$ where $g(\cdot)$ is some
  stochastic function, 
\item some unknown (and noisy) combination of $X_i: Y_i = g(X_{i_1}, X_{i_2},
  \ldots)$, or 
\item new attributes that are independent of any of the $X_i$. 
\end{itemize}
Further, some of the $X_i$ might even be omitted entirely when generating $Y_i$. 

Our goal is (ideally) to determine $o$.
We instead design an algorithm that learns $\tilde{o} = \tilde{f}(Y_1, Y_2, \ldots
Y_\ell)$ i.e a mapping from \obsspace to \decisionspace. The  hope is that
$\tilde{o} \simeq o$. 

\subsection{Examples}
\label{sec:examples}

The easiest way to understand the interactions between the \acl{ds}, the
\acl{ts} and the \acl{os} is to start with a prediction task, posit features
that seem to control the prediction, and then imagine ways of \emph{measuring}
these features. We provide a number of such examples in Table~\ref{tab:tsos}, described in more detail below.

\begin{table}[htbp]
  \centering
\begin{tabular}{ccc}
\toprule
\Acl{ds} & \Acl{ts} & \Acl{os} \\
\midrule
Performance in college &Intelligence & IQ \\
Performance in college &Success in High School & GPA\\
Recividism&Propensity to commit crime & Family history of crime\\
Recidivism&Risk-averseness & Age\\
Employee Productivity &Knowledge of job & Number of Years of Experience\\
\bottomrule
\end{tabular}
\caption{Examples of \acl{ts} attributes and their corresponding \acl{os}
  attributes for different outcomes\label{tab:tsos}}
\end{table}

\textbf{College Admission.}
Universities consider a number of factors when deciding who to admit to their
university. One admissions goal might be to determine the likelihood that an admitted student
will be successful in college, and the factors considered can include things like intelligence, high school
GPA, scores on standardized tests, extracurricular activities and so on.  In this example, 
performance in college is the \acl{ds}, while intelligence and success in high school are in the \acl{ts}.    Intelligence might be represented in the \acl{os} by the result of an IQ test, while success in high school could be observed by high school GPA.

\textbf{Recidivism Prediction.} When an offender is eligible for
parole, judges assess the likelihood that the offender will re-offend after
being released as part of the parole decision.  Many jurisdictions now use automated prediction methods like
\textsc{COMPAS}\cite{compas} to generate such a
likelihood. With the goal of predicting the likelihood of recidivism (the \acl{ds}), such an algorithm might want to determine an individual's propensity for criminal activity and their level of risk-aversion.  These \acl{ts} attributes could be modeled in the \acl{os} by a family history of crime and the offender's age.

\textbf{Hiring.} One of the most important criteria in
hiring a new employee is their ability to succeed at a future job. As proxies,
employers will use features like the college attended, GPA, previous work
experience, interview performance and their overall resume.  The \acl{ds} in this case is employee productivity once hired, while a \acl{ts} attribute is the applicant's current knowledge of the job.  One way to observe this knowledge (an attribute in the \acl{os}) is by their number of years of experience at a similar job.

\subsection{Quantifying transformations between spaces}
\label{sec:comparing-spaces}

We can describe the entire pipeline of algorithmic decision-making (feature
extraction and measurement, prediction algorithms and even the underlying
predictive mechanism) in the form of transformations between spaces. 
This is where the \emph{metric} structure of the spaces plays a role. As we will
see, we can express the quality of the various transformations between spaces in
terms of how distances (that capture dissimilarity between entities)
change when one space is transformed into another. The reason to use (functions
of) distances to
compare spaces is because most learning tasks rely heavily on the underlying distance
geometry of the space. By measuring how the distances change relative to
the original space we can get a sense of how the task outcomes might be
distorted.

We introduce two different approaches to quantifying these transformations: one
that is more ``local'' and one that compares how \emph{sets} of points are
transformed. We start with a standard measure of point-wise transformation
cost. 

\begin{definition}[(additive) Distortion]
Let $(X, d_X)$ and $(Y, d_Y)$ be two metric spaces and let $f : X \to Y$ be a map from $X$ to $Y$. The \emph{distortion} $\rho_f$ of $f$ is defined as the smallest value such that for all $p, q\in X$
\[ |d_X(p,q) -  d_Y(f(p), f(q))| \le \rho_f\]
The distortion $\rho(X, Y)$ is then the minimum achievable distortion $\rho_f$
over all mappings $f$. 
\end{definition}

\paragraph{Notes.}
While the above notion (and its multiplicative variant) is standard in the theoretical computer science
literature, it is helpful to understand why it is justified in the context of
algorithmic fairness. Distortion is commonly used as way to minimize the change of
geometry when doing dimensionality reduction to make a task more efficiently
solvable \footnote{It is possible to take a more information-theoretic perspective on the nature of
transformations between spaces. For example, we could quantify the quality of a
(stochastic) transformation by the amount of mutual information between the
source and target spaces. While this is relevant when we wish to determine the
degree to which we can reconstruct the source space from the target, it is not
necessary for algorithmic decision-making. For example, it is not necessary that
we be able to reconstruct the \acl{ts} features from features in the
\acl{os}. But it \emph{will} matter that individuals with similar features in the
\acl{ts} have similar features in the \acl{os}.}. The specific measure above is
a special case $(p = \infty)$ of a general $\ell_p$-additive distortion (where
the norm is computed over the vector of distance differences). In this more
general setting there are a number of approximation algorithms for estimating
$\rho$ when the ``target space'' $Y$ is restricted to a line \cite{dhamdhere2004approximating}, a tree\cite{agarwala1998approximability}, or an ultrametric\cite{farach1995robust}.

There are many different ways to compare metric spaces using their
distances. Distortion is a worst-case notion: it is controlled by the worst-case
spread between a pair of distances in the two spaces. If instead we wished to
measure distances between subsets of points in a metric space, there is a more
appropriate notion.

\begin{definition}[Coupling Measure]
Let $X, Y$ be sets with associated probability measures $\mu_X, \mu_Y$. A
probability measure $\nu$ over $X \times Y$ is a \emph{coupling measure} if
$\nu(X, \cdot)$ (the projection of $\nu$ on $X$) equals $\mu_X$, and similarly
for $\nu(\cdot, Y)$ and $\mu_Y$. The space of all such coupling measures is
denoted by $\mathcal{U}(X, Y)$.
\end{definition}

\begin{definition}[\Ac*{wd}]
  Let $(X, d)$ be a metric space and let $Y, Y'$ be two subsets of $X$. Let
  $\mu$ be a probability measure defined on $X$, which in turn induces
  probability measures $\mu_Y, \mu_{Y'}$ on $Y, Y'$ respectively. 

  The \acl{wd} between $Y, Y'$ is given by 
\[ \emd_d(Y, Y') = \min_{\nu \in \mathcal{U}(Y, Y')]} \int d(y, y') \nu(y, y') \]
\end{definition}

The \ac{wd} finds an optimal transportation between the two sets and computes
the resulting distance. It is a metric when $d$ is. 

Finally, we need a metric to compare subsets of points that lie in \emph{different} metric
spaces. Intuitively, we would like some distance function that determines
whether the two subsets have the same \emph{shape} with respect to the two
underlying metrics. We will make use of a distance function called the
\emph{Gromov-Wasserstein distance} \cite{gw} that is derived from the \acl{wd}
above.

\begin{definition}[\Ac*{gw}]
  Let $(X, d_X), (Y, d_Y)$ be two metric spaces with associated probability
  measures $\mu_X, \mu_Y$. The \acl{gw} between $X$ and $Y$ is given by 
  \[ \gw(X, Y) = \frac{1}{2} \inf_{\nu \in \mathcal{U}(X, Y)} \int \int |d_X(x, x')
  - d_Y(y, y')|d\mu_X\times d\mu_X d\mu_Y\times d\mu_Y \]
\end{definition}

Intuitively, the \ac{gw} computes the \ac{wd} between the sets of \emph{pairs}
of points, to determine whether the two point sets determine similar sets of
distances. 

We note that both $\emd(X,Y)$ and $\gw(X,Y)$ can be computed using the
standard Hungarian algorithm for optimal transport. $\emd(X,Y)$ can be computed
in time $O(n^3)$ (where $|X| = |Y| = n$) and $\gw(X,Y)$ can be computed in time
$O(n^6)$. 

%% file: defns.tex
\section{A mathematical formulation of fairness and bias}
\label{sec:definitions}

We have now introduced three spaces that play a role in algorithmic decision-making: the \acl{ts}, the \acl{os} and the \acl{ds}. We have also introduced ways to measure the fidelity with which spaces map to each other. Armed with these ideas, we can now describe how notions of fairness and bias can be expressed formally. 

\subsection{A definition of fairness}

The definition of fairness is task specific, and prescribes desirable outcomes for a task. Since the solution to a  task \emph{is a mapping from the \acl{ts} \truespace to the \acl{ds} \decisionspace}, a definition of fairness should describe the properties of such a  mapping. Inspired by the fairness definition due to Dwork et al.\cite{Dwork12Fairness}, we give the following definition of fairness:

\begin{definition}[Fairness]
  A mapping $f : \truespace \to \decisionspace$ is said to be \emph{fair} if objects that are close in $\truespace$ are also close in \decisionspace. Specifically, fix two thresholds $\epsilon, \epsilon'$. Then $f$ is defined as $(\epsilon, \epsilon')$- fair if for any $x, y \in P$, 
\[ d_P(x,y) \le \epsilon \implies d_O(f(x), f(y)) \le \epsilon' \]
\end{definition}
Note that the definition of fairness does not require any particular outcome for entities that are far apart in \truespace. 

\subsection{A worldview: what you see is what you get}
The presence of the \acl{os} complicates claims that data-driven decision making can be fair, since features in the observed space might not reflect the true value of the attributes that you would like to use to make the decision.  In order to address this complication, given that the \acl{ts} is unobservable, assumptions must be introduced about the points in the \acl{ts}, or the mapping between the \acl{ts} and \acl{os}, or both.

One worldview focuses on the mapping between the \acl{ts} and \acl{os} by asserting that the \acl{ts} and \acl{os} are \emph{essentially the same}. We call this worldview the \ac{wysiwyg} view. 

\begin{axiom}[\ac{wysiwyg}]
There exists a mapping $f \colon \truespace \to \obsspace$ such that the distortion $\rho_f$ is at most $\epsilon$ for some small $\epsilon > 0$. Or equivalently, the distortion $\rho$ between \truespace and \obsspace is at most $\epsilon$
\end{axiom}

In practice, we can think of $\epsilon$ as a very small number like $0.01$. 

\bignote{\textbf{Example: \acs{wysiwyg} in a College Admissions Decision}\\
In the college admissions setting, the \acs{wysiwyg} is the assumption that features like SAT scores and high-school GPA (which are observed) correlate well with the applicant's ability to succeed (a property of the \acl{ts}). More precisely, it assumes that there is some way to use a combination of these scores to correctly compare true applicant ability. 
}

\subsection{A worldview: structural bias}
\label{sec:systemic-bias}

But what if the \acl{ts} isn't accurately represented by the \acl{os}?  In the case of stochastic noise in the transformation between \acs{ts} and \acs{os}, fairness in the system may decrease for all decisions.  This case can be handled using the \acs{wysiwyg} worldview and usual techniques for accurate learning in the face of noise (see, e.g., \cite{Kearns1998Noise}).

Unfortunately, in many real-world societal applications, the noise in this transformation is non-uniform in a societally biased way.  To explain this \emph{structural bias}, we start with the notion of a \emph{group}: a collection of individuals that share a certain set of characteristics (such as gender, race, religion and so on).  These characteristics are often historically and culturally defined (e.g., by the long history of racism in the United States). We represent groups as a partition of individuals into sets $G_1, G_2, \ldots, G_k$. In this work, we will think of a group membership as a characteristic of an individual; thus each of the \acl{ts}, \acl{os}, and \acl{ds} admits a partition into groups, induced by the group memberships of individuals represented in these spaces. 

Structural bias manifests itself in unequal treatment of groups. In order to quantify this notion, we first define the notion of \emph{group skew}: the way in which group (geometric) structure might be distorted between spaces. What we wish to capture is the \emph{relative} distortion of groups with respect to each other, rather than (for example) a scaling transformation that would transform all groups the same way. 

Let $(X, d)$ be a metric space partitioned into groups $\mathcal{X} = \{X_1, \ldots, X_k\}$. Any probability measure $\mu_X$ defined on $X$ induces a measure $\mu_{\mathcal{X}}$ on $\mathcal{X}$ in the natural way. We can define a metric $d_{\mathcal{X}}$ on $\mathcal{X}$ via the operation $d_{\mathcal{X}}(X_i, X_j) = \emd_d(X_i, X_j)$. Now consider two such metric spaces $(X, d_X)$, $(Y, d_Y)$ and their associated group metric spaces $(\mathcal{X}, d_{\mathcal{X}}), (\mathcal{Y}, d_{\mathcal{Y}})$ and measures $\mu_X, \mu_Y$ 

\begin{definition}[Between-groups distance]
The \emph{between-groups distance} between $(X, d_X)$, $(Y, d_Y)$ with measures $\mu_X, \mu_Y$ is 
\[ \rho_b = \frac{\gw(\mathcal{X}, \mathcal{Y})}{\binom{k}{2}} \]
\end{definition}

The between-groups distance treats the groups in a space as individual ``points'', and compares two collections of ``points''. To capture the \emph{differential} treatment of groups, we need to normalize this against a measure of how each group is distorted individually\footnote{This is similar to how we might measure between-group and within-group variance in statistical estimation problems like ANOVA.}. 

\begin{definition}[Within-group distance]
Let $X_i$ and $Y_i$ be the two sets in the spaces $X, Y$ corresponding to the $i^{th}$ group. Let $\rho_i = \gw(X_i, Y_i)$. Then we define 
\[ \rho_w = \frac{1}{k}\sum_{i=1}^k \rho_i \]
\end{definition}

We can now define a notion of \emph{group skew} between two spaces. 

\begin{definition}[Group skew]
Let $(X, d_X)$ and $(Y, d_Y)$ be metric spaces with group partitioning $\mathcal{X}, \mathcal{Y}$ and measures $\mu_X, \mu_Y$. The \emph{group skew} between $\mathcal{X}$ and $\mathcal{Y}$ is the quantity 
\[ \gs(\mathcal{X}, \mathcal{Y}) = \frac{\rho_b(\mathcal{X}, \mathcal{Y})}{\rho_w(\mathcal{X}, \mathcal{Y})} \]
\end{definition}

There is a degenerate case in which group skew is not well-defined. This is when for each $i$ the sets $X_i$ and $Y_i$ are identical in distance structure. In this (admittedly unlikely) setting, each $\rho_i$ will be zero, and thus $\rho_w = 0$. This can be interpreted as saying that when groups are identical in the two spaces, any small variation between groups is magnified greatly. To avoid this degenerate case, we will instead compute $\rho_b$ and $\rho_w$ on a \emph{perturbed} version of the data, where each point is shifted randomly within a metric ball of radius $\delta$. The parameter $\delta$ acts as a \emph{smoothing operator} to avoid such degenerate cases. This effectively adds $O(\delta)$ to each of the numerator and denominator, ensuring that the ratio is always well defined.

Using these definitions, we can now account for structural bias, which can be informally understood as the existence of more distortion between groups than there is within groups when mapping between the \acl{ts} and the \acl{os}, thus identifying when groups are treated differentially by the observation process.

\begin{definition}[Structural Bias]
The metric spaces \acs{ts} $= (X, d_X)$ and \acs{os} $= (Y, d_Y)$ admit $t$-\emph{structural bias} if the group skew $\gs(\mathcal{X}, \mathcal{Y}) > t$.
\end{definition}

\bignote{\textbf{Example: Structural Bias in a College Admissions Decision}\\
Researchers have shown that the SAT verbal questions function differently for the African-American subgroup, so that the validity of the results as a measure of ability are in question for this subgroup \cite{Santelices10UnfairSAT}.  In the case where SAT scores are a feature in the \acl{os}, this research indicates that we should consider these scores to be the result of structural bias.
}

\subsubsection{Non-Discrimination: a top-level goal}
Since group skew is a property of two metric spaces, we can consider the impact of group skew between the \acl{ts} and \acl{os} (structural bias as defined above), between the \acl{os} and the \acl{ds}, and between the \acl{ts} and the \acl{ds}.  While colloquially ``structural bias" can refer to any of these (since the \acl{ts} and \acl{os} are often conflated), in this paper we will give different names to group skew depending on the relevant spaces used in the mapping.  We will refer to group skew in the decision-making procedure (the mapping from \acl{os} to \acl{ds}) as \emph{direct discrimination}.

\begin{definition}[Direct Discrimination]
The metric spaces  $\obsspace = (X, d_X)$ and  $\decisionspace= (Y, d_Y)$ admit $t$-\emph{direct discrimination} if the group skew $\gs(\mathcal{X},\mathcal{Y}) > t$.
\end{definition}
Note that the group structure $\mathcal{Y}$ is the direct result of a mapping $f : \obsspace \to \decisionspace$, so we can think of direct discrimination as a function of this mapping. 

Since group membership is usually defined based on innate or culturally defined characteristics that individuals have no ability to change, it is often considered unacceptable (and in some cases, illegal) to use group membership as part of a decision-making process.  Thus, in decision-making \emph{non-discrimination} is often a high-level goal.  This is sometimes termed ``fairness," but we will distinguish the terms here.

\begin{definition}[Non-Discrimination]
\label{def:non-discrim}
Let  $\truespace= (X, d_X)$ and  $\decisionspace= (Y, d_Y)$. A mapping $f : \truespace \to\decisionspace$ is $t$-nondiscriminatory if the group skew $\rho(\mathcal{X}, \mathcal{Y}) \le t$.

\end{definition}

Thus, this worldview is primarily concerned with achieving non-discrimination by avoiding both structural bias and direct discrimination.  Given the social history of this type of group skew occurring in a way that disadvantages specific sub-populations it makes sense that this is a common top-level goal.  Unfortunately, it is hard to achieve directly, since we have no knowledge of the \acl{ts} and the existence of structural bias precludes us from using the \acl{os} as a reasonable representation of the \acl{ts} (as is done in the \acs{wysiwyg} worldview).

\subsubsection{An Axiomatic Assumption: we're all equal}
Instead, a common underlying assumption of this worldview, that we will make precise here, is that in the \acl{ts} \emph{all groups look essentially the same}. In other words, it asserts that there are no innate differences between groups of individuals defined via certain potentially discriminatory characteristics.  This latter axiom of fairness appears implicitly in much of the literature on statistical discrimination and disparate impact. 

There is an alternate interpretation of this axiom: the groups aren't equal, but for the purposes of the decision-making process they should be treated as if they were.  In this interpretation, the idea is that any difference in the groups' performance (e.g., academic achievement) is due to factors outside their individual control (e.g., the quality of their neighborhood school) and should not be taken into account in the decision making process \cite{Roemer}.  This interpretation has the same mathematical outcome as if the equality of groups is assumed as true, and thus we will refer to a single axiom to cover these two interpretations.

\begin{axiom}[\acl{wae} (\acs{wae})]
Let $\truespace = (X, d_X)$ with measure $\mu_X$ be partitioned into groups $X_1, \ldots, X_k$. 
There exists some $\eps > 0$ such that for all $i, j$,  $\emd_{d_X}(X_i, X_j) < \eps$. 
\end{axiom}

\bignote{\textbf{Example: \acs{wae} in a College Admissions Decision} \\
In the college admissions setting, the \acs{wae} asserts that all groups will have almost the same distribution in the \acl{ts} of intrinsic abilities, such as grit or intelligence, chosen as important inputs to the decision making process.  In the example of SAT scores, given above, this would mean that we assume the structural bias of these scores apparent in the \acl{os} is not representative of a distributional difference in the \acl{ts}.
}

It is useful to note that the \acs{wae} is a property of the \acl{ts}, whereas the \acs{wysiwyg} describes the relation between the \acl{ts} and \acl{os}.  Note also that the definition of structural bias does not itself assume the \acs{wae} - in fact, there could be structural bias that acted in addition to existing true differences between groups present in the \acl{ts} to further separate the groups in the \acl{os}.  However, because of the lack of knowledge about the \acl{ts} when assuming the existence of structural bias, the \acs{wae} will often be assumed in practice under the structural bias worldview.

\subsection{Comparing Worldviews}
\label{sec:comparing-worldviews}

While we introduce these two axioms as different world views or belief systems, they can also be strategic choices.  Whatever the motivation (which is ultimately mathematically irrelevant), the choice in axiom is critical to a decision-making process.  The chosen axiom determines what fairness means by giving enough structure to the \acl{ts} or the mapping between the \acl{ts} and \acl{os} to enforce fairness despite a lack of knowledge of the \acl{ts}.  We discuss the subtleties of the axiomatic choice here and will return to the enforcement of fairness based on this axiomatic choice in the next section.

\begin{figure}[htb]
\begin{center}
\bignote{\textbf{Example: Choice of \acs{wysiwyg} or \acs{wae} for a College Admissions Decision}\\
There are many reasonable ways that a college admissions office might decide to set up a fair decision making procedure.  These would include different choices of what features to be included in the \acl{ts} and might indicate different fairness goals.  They would also necessitate different axiomatic choices.
\begin{enumerate}
\item One decision-making philosophy might be that the college should admit only those students who reach a high level of achievement and demonstrated intelligence at the time of admission.  The \acl{ts} in this example might include potential, intelligence, and diligence \emph{at the time of the admissions decision}.
	\begin{enumerate}
	\item If the admissions office believes that their \acl{os} features accurately represent the \acl{ts} features, this scenario aligns with the \acs{wysiwyg} axiom.
	\item If the admissions office believes that any systemic group differences in the \acl{os} are inaccuracies (e.g., potentially due to culturally inflected exam questions), this scenario follows the \acs{wae} axiom.
	\end{enumerate}
\item Another decision-making philosophy might be that the college should focus on admitting those with high innate potential, regardless of the social environment and life experiences that may have shaped that potential.  In this case, the construct space might include potential, intelligence, and diligence \emph{at birth}.  As above, the admissions office could choose to believe either the \acs{wysiwyg} or the \acs{wae} axioms.
\item A third decision-making philosophy might be that the college admissions process should serve as a social equalizer, so that, e.g., applicants from different class backgrounds are admitted at approximately the same rate.  Since the \acl{ts} is the space of features used for ideal decision-making, in this case potential, intelligence, and diligence might be assumed to represent an idealized belief of the characteristics and abilities of an individual were their class background equalized.  (Some may believe that this case is the same as representing these qualities at the time of birth.)  This would be a choice to follow the \acs{wae} axiom.
\end{enumerate}
}
\end{center}
\end{figure}

The choice of worldview is heavily dependent on the specific attributes and task considered, and on the algorithm designer's beliefs about how observations of these attributes align with the associated constructs.  Roemer identifies the goal of such choices as ensuring that negative attributes that are due to an individual's circumstances of birth or to random chance should not be held against them, while individuals should be held accountable for their effort and choices \cite{Roemer}.  He suggests that differences in worldview can be attributable to \emph{when} in an individual's development the playing field should be leveled and after what point an individual's own choices and effort should be taken into account.  In our decision-making formulation, the decision about when amounts to a decision of which axiom to believe at the point in time the decision will be made.  If the decision is being made while the playing field should be leveled, then the \acl{wae} axiom should be assumed.  If the decision is being made while only an individual's own efforts should be included in the decision, then the \acs{wysiwyg} axiom may be the right choice.

\subsection{Mechanisms}
\label{sec:policies}

We can think of the axioms as assumed relationships between the \acl{ts} and the \acl{os} (or operating within the \acl{ts}), and fairness definitions as desirable outcomes (executions of tasks) that reflect these relationships. A \emph{mechanism} is then a constructive expression of a definition: it is a mapping (or set of mappings) from \obsspace to \decisionspace that allow the definition to be satisfied. In effect, a well-designed mechanism working from a specific set of axioms should yield a fair outcome. 

Formally, a \emph{mechanism} is a mapping $f : \obsspace \to \decisionspace$ that satisfies certain properties. 
First and foremost, a mechanism should be \emph{nontrivial}. For example, if the
decision space is $\{0,1\}$ (e.g for binary classification), a mechanism that
assigned a $0$ to each point would be trivial. 

\begin{definition}[\Acl{rich}]
A mechanism $f : \obsspace \to \decisionspace$ is \emph{\ac*{rich}} if for each $d \in
\decisionspace$, $f^{-1}(d) \ne \emptyset$. 
\end{definition}

There are then two types of mechanisms that (we will show) provide guarantees under the two different world views described above:  \acl{ifm}s (aiming to guarantee fairness) and \acl{gfm}s (aiming to guarantee non-discrimination).

\begin{definition}[\Acl{ifm} (\acs{ifm})]
Fix a tolerance $\epsilon$. A mechanism \acs{ifm}$_\epsilon$ is a \ac*{rich} mapping $f : \obsspace \to \decisionspace$ such that $\rho_f \le \epsilon$.
\end{definition}

The \acl{ifm} asserts that the mechanism for decision making treats people similarly if they are close, and can treat them differently if they are far, in the observed space. 

\begin{definition}[\Acl{gfm} (\acs{gfm})]
  Let $X$ be partitioned into groups $X_1, X_2, \ldots$ as before. A \ac*{rich} mapping $f: $ \acs{os}$ \to $ \acs{ds} is said to be a valid \emph{\acl{gfm}} if all groups are treated equally. Specifically, fix $\epsilon$. Then $f$ is said to be a \acs{gfm}$_\epsilon$ if for any $i, j$, $\emd_{d_O}(X_i, X_j) \le \epsilon$.
\end{definition}

The \acl{gfm} asserts that the decision mechanism should treat all groups the same
way. The doctrine of disparate impact is an example of such an assertion
(although the precise measure of the \emph{degree} of disparate impact as
measured by the $4/5$-rule is different). 

In the following sections we will explore whether and how these types of mechanisms actually guarantee fairness under certain axiomatic assumptions.


%% file: impos.tex
\section{Making fair or non-discriminatory decisions}

With the basic vocabulary in place, we can now ask questions about when fairness is possible. An easy first observation is that under the \acs{wysiwyg}, we can always be fair. 

\begin{theorem}
\label{thm:ifm}
  Under the \acs{wysiwyg} with error parameter $\delta$, an \acs{ifm}$_{\delta'}$ will guarantee $(\epsilon, \epsilon')$-fairness for some function $f$ such that $\epsilon' = f(\delta, \delta')$. 
\end{theorem}

\begin{proof}
  Two points in the \acl{os} at distance $d$ have a distance in the \acl{ts} between $d-\delta$ and $d+\delta$. Applying the mechanism $\ip_{\delta'}$ yields decision elements that have a distance (in \decisionspace) between $d-\delta-\delta'$ and $d+\delta+\delta'$. Setting $\epsilon'$ appropriately yields the claim. 
\end{proof}

The requirement that we have an individually fair mechanism turns out to be important. 
We start with some background. Fix $\truespace$. Assume that we have two groups $a, b$, and so each
point $p \in \truespace$ has a label $\ell(p) : \truespace \to \{a, b\}$. Let
$\phi : \truespace \to \obsspace$ be the method by which features of $p$ are
``observed''. For simplicity, we assume this map is bijective and so for each $q
\in \obsspace$ there exists $\phi^{-1}(q) = p \in \truespace$. We will abuse
notation and denote the (group) label of $q \in \obsspace$ as $\ell(q) =
\ell(\phi^{-1}(q))$. Let $d_{\obsspace}$ be a metric on $\obsspace$. Let $P
\subset \obsspace$ be a set of points. The \emph{diameter} of $P$ is $\Delta(P) = \max_{x,y\in P} d_{\obsspace}(x,y)$. 
Consider an arbitrary  $f$. 

\begin{theorem}
\label{thm:decisions}
Under the \acs{wysiwyg} with parameter $\epsilon$, for any $\delta, \delta' < 1$ and a rich mechanism $f : \obsspace \to
\decisionspace$ where the \acl{ds} is discrete ($\decisionspace = \{0,1, 2, \ldots, k\}$), $f$ is not $(\delta-\epsilon,\delta')$-fair. 
\end{theorem}
\begin{proof}
Fix the metric $d(x,y) =  \mathbf{1}_{x \ne y}$.
Let $\mathcal{B} = B_r(x)$ be a ball of radius $r$ centered at $x \in \obsspace$. We will say
that $\mathcal{B}$ is \emph{monochromatic} if all points in $\mathcal{B}$ have
the same image under $f$. 
Let $r'$ be the smallest value of $r$ such that $B_r(x)$ is not
monochromatic. If such an $r'$ does not exist, then $f$ cannot be rich. Consider
the difference $\Delta = B_r(x) \setminus B_{r-\delta}(x)$, and let
$B_{\delta/2}(y) \in \Delta$ be some ball that is not monochromatic (such a ball must exist since $f$ is injective). 
Pick two points $p,q \in B_{\delta/2}(y)$
that have different images under $f$. But they are at most $\delta$ apart! Any
bijection $\phi$ from $\truespace$ to $\obsspace$ that preserves this distance will thus
ensure that there are two points $\phi^{-1}(p)$ and $\phi^{-1}(q)$ that are
within distance $\delta - \epsilon$ but have a distance of $1 > \delta'$ in $\decisionspace$.
\end{proof}

The essence of the above argument is that a discrete decision space disallows a fair mechanism, and precludes fairness. 
\subsection{Non-discriminatory decisions are possible}

Demographic parity, the disparate impact four-fifths rule, and other measures quantifying the similarity in the outcomes that groups receive in the \acl{ds} are prevalent and associated with many \acl{gfm}s that attempt to guarantee good outcomes under these measures. We will show that such \acl{gfm}s guarantee non-discriminatory decisions.

Recall from Definition \ref{def:non-discrim} that non-discriminatory decisions guarantee a lack of group skew in the mapping between the \acl{ts} and \acl{ds}, i.e., the goal of non-discrimination is to ensure that the \emph{process} of decision-making does not vary based on group membership.  Given that it's not possible to directly measure the \acl{ts} or the mapping between the \acl{ts} and \acl{ds} without assuming the \acs{wysiwyg} axiom, these \acl{gfm}s attempt to ensure non-discrimination through measurements of the \acl{ds}.

Do these \acl{gfm}s succeed?  Not at first glance. Suppose we have two groups $X_1, X_2$ in \truespace that are far apart, i.e $\emd(X_1, X_2)$ is large and suppose also that they are appropriately far apart in their performance on the task. Suppose that because of structural bias, the images of these two groups in the \acl{os} \obsspace are even further apart while keeping the distribution of task performance within each group the same. 

 A \acl{gfm} applied to the \acl{os} will then move these groups, on the whole, to the same smaller portion of the \acl{ds} so that they receive decisions indicating that they are, on the whole, equal with respect to the task.  (Suppose again that the individuals within the group are mapped similarly with respect to each other and the task.)  

Is this decision process non-discriminatory?  No.  While the within-group distortion will remain the same between the \acl{ts} and the \acl{ds}, the between-group distortion will be as large as the separation between $X_1$ and $X_2$ in the \acl{ts}.  Intuitively, we can see this as discriminatory towards the group that performs better with respect to the task in the \acl{ts}, since they are, as a group, receiving worse decisions than less skilled members of the other group (i.e., there has been group skew in their group's mapping to the decision space).

Yet these \acl{gfm}s are in common practice -- why?  First, let's review the assumptions of this scenario.  If the \acs{wysiwyg} axiom is assumed, then guaranteeing fairness is easily achievable, so here we are interested in what to do in the case when the \acs{wysiwyg} axiom is \emph{not} assumed.  Specifically, let's assume that we are worried about the existence of structural bias -- group skew in the mapping between the \acl{ts} and the \acl{os}.  In this scenario, it may make sense to assume the \acl{wae} axiom.  In fact, as we will show now, when the \acl{wae} axiom is assumed \acl{gfm}s can be shown to guarantee non-discrimination.

\begin{theorem}[\Acl{gfm}s guarantee non-discrimination]
Under the \acs{wae}, a \acs{gfm} with parameter $\epsilon'$ guarantees $(\max(\epsilon,\epsilon')/\delta)$-nondiscrimination. 
\end{theorem}

\begin{proof}
The \wae ensures that in the \acl{ts} \truespace, all groups are within distance $\epsilon$ from each other under $W_d$. Similarly, a \acl{gfm} ensures that in the \acl{ds} \decisionspace, all groups are within distance $\epsilon'$ from each other. 

Consider now the between-group distance $\rho_b$ between $\truespace$ and \decisionspace. Since all groups are within $\epsilon$ of each other in \truespace and within $\epsilon'$ in \decisionspace, each term in the integral that computes $\gw$ is upper bounded by $\max(\epsilon, \epsilon')$. By construction, the within-group distance $\rho_w$ is lower bounded by the noise parameter $\delta$. Thus, the overall structural bias score $\sigma$ is upper bounded by $\max(\epsilon,\epsilon')/\delta$. 
\end{proof}

Note that this guarantee of non-discrimination holds even under the structural bias worldview, i.e., the theorem makes no assumptions about the mapping from the \acl{ts} to the \acl{os}.  With this theorem, we now have both an axiom and mechanism under which fairness can be achieved, and a corresponding axiomatic assumption and mechanism under which non-discrimination can be achieved.  

\subsection{Conflicting worldviews necessitate different mechanisms}

As we have shown in this section, under the \acs{wysiwyg} worldview fairness can be guaranteed, while under a structural bias worldview non-discrimination can be guaranteed.  Are these worldviews fundamentally conflicting, or do mechanisms exist that can guarantee fairness or non-discrimination under both worldviews?

Unfortunately, as discussed above, the \acs{wysiwyg} appears to be crucial to ensuring fairness: if for example there is structural bias in the decision pipeline, no mechanism can guarantee fairness. 
Fairness can only be achieved under the \acs{wysiwyg} worldview using an \acl{ifm}, and using a \acl{gfm} will be \emph{unfair} within this worldview.

What about non-discrimination?  Unfortunately, a simple counterexample again shows that these mechanisms are not agnostic to worldview.  While \acl{gfm}s were shown to achieve non-discrimination under a structural bias worldview and the \acl{wae} axiom, if structural bias is assumed, applying an \acl{ifm} will cause \emph{discrimination} in the \acl{ds} whether the \acl{wae} axiom is assumed \emph{or not}.  Consider again the two groups $X_1, X_2$ in \truespace with large $\emd(X_1, X_2)$, and again suppose that the images of these two groups in the \acl{os} \obsspace are even further apart, while keeping the distribution of task performance in each group the same.  Now apply an \acl{ifm} to this \acl{os}.  The resulting \acl{ds} contains a large between-group distortion since the group that performed better with respect to the task in the \acl{ts} will have received, on the whole, much better decisions than their original skill with respect to the other group warrants.  These decisions will thus be discriminatory.

Choice in mechanism must thus be tied to an explicit choice in worldview.  Under a \acs{wysiwyg} worldview, only \acl{ifm}s achieve fairness (and \acl{gfm}s are unfair).  Under a structural bias worldview, only \acl{gfm}s achieve non-discrimination (and \acl{ifm}s are discriminatory).


%% file: related.tex
\section{Analyzing Related Work}
\label{sec:related}

This section serves partly as a review of the literature on fairness. But it also serves as a form of ``empirical validation'' of our framework, in that we use our new formalization of what fairness and non-discrimination mean and the underlying assumptions necessitated when attempting to build fair mechanisms in order to  reconsider previous work within this framework.  Broadly, we find that the previous work in fairness-aware algorithms either
\begin{inlinelist}
\item adopt the \acs{wysiwyg} worldview and guarantee fairness while assuming the \acs{wysiwyg} axiom or
\item adopt the structural bias worldview and guarantee non-discrimination while assuming the \acl{wae} axiom.
\end{inlinelist}
A full survey of such work can be found in \cite{Romei13Multidisciplinary,zliobaite2015survey}.  Here, we will describe some interesting representative works from each of the worldviews.

\subsection{WYSIWYG Worldview}
One foundational work that adopts the \acs{wysiwyg} worldview is Dwork et al. \cite{Dwork12Fairness}.  The definition of fairness they introduce is similar to (and inspired) ours -- they are interested in ensuring that two individuals who are similar receive similar outcomes.  The difference from our definition is that they consider outcome similarity according to a distribution of outcomes for a specific individual.  Dwork et al. emphasize that determining whether two individuals are similar with respect to the task is critical, and assume that such a metric is given to them.  In light of the formalization of the \acl{ts} and \acl{os}, we add the understanding that the metric discussed by Dwork et al. is the distance in the \acl{ts}.  In our framework, this metric is not knowable unless \acs{wysiwyg} is assumed (or the specific mapping between the \acl{ts} and \acl{os} is otherwise provided), so we classify this work as adopting the \acs{wysiwyg} worldview.

Additionally, Dwork et al. \cite{Dwork12Fairness} show that when the earthmover distance between distributions of the attributes conditioned on protected class status are small, then their notion of fairness implies non-discrimination (which they measure as \emph{statistical parity}, or a ratio of one between the positive outcome probability for the protected class and that for the non-protected class).  Thus, they show that under an assumption similar to the \acs{wysiwyg} axiom, if an assumption similar to the \acl{wae} axiom is also assumed, then \acl{gfm}s guarantee fairness.  Note that this special case is unusual, since both axiomatic assumptions are made.  A follow-up work by Zemel et al. \cite{icml2013_zemel13} attempts to bridge the gap between these worldviews by adding a regularization term to attempt to enforce statistical parity as well as fairness.

Interestingly, some of the examples in Dwork et al. \cite{Dwork12Fairness} arguing that a particular form of non-discrimination measure (``statistical parity") is insufficient in guaranteeing fairness make an additional subtle assumption about what spaces are involved in the decision-making process.  Their model implicitly assumes that there could be both an observed decision space and a true decision space (a scenario common in the differential privacy literature), while our framework assumes only a single truly observable decision space (as is more common in the machine learning literature).  One example issue they introduce is the ``self-fulfilling prophecy" in which, for example, an employer purposefully brings in under-qualified minority candidates for interviews (the observed decision space) so that no discrimination is found at the interview stage, but since the candidates were under-qualified, only white applicants are eventually hired (the true decision space).  Under our framework, only the final decisions about who to hire make up the single decision space, and so the discrimination in the decision is detected.

Another type of fairness definition is based on the amount of change in an algorithm's decisions when the input or training data is changed.  Datta et al. \cite{Datta2015AdPrivacy} consider ad display choices to be discriminatory if changing the protected class status of an individual changes what ads they are shown.  Fish et al. \cite{Fish2016ConfidenceBased} consider a machine learning algorithm to be fair if it can reconstruct the original labels of training data when noise has been added to the labels for anyone from a given protected class.  Both of these definitions make the implicit assumption that the remaining training data that is not the protected class status or the label is the correct data to use to make the decision.  This is exactly the \acs{wysiwyg} axiomatic assumption.

A recent work by Joseph et al. \cite{Joseph2016Bandits} also contributes a new fairness definition, akin to those introduced in this paper and by Dwork et al., that aims to ensure that worse candidates are never accepted over better candidates as measured with respect to the task.  Their goal is to take these measurements within the \acl{ts} with unknown per-group functions mapping from the \acl{ts} to the \acl{os}.  Joseph et al. aim to learn these per-group functions.  Thus, although their fairness goal focuses on fairness at an individual level, this work serves as a bridge to the structural bias worldview by recognizing that different groups may receive different mappings between the \acl{ts} and the \acl{os}.

\subsection{Structural Bias Worldview}

The field of fairness-aware data mining began with examinations of how to ensure non-discrimination in the face of structural bias.  These \acl{gfm}s often implicitly assume the \acl{wae} axiom and, broadly, share the goal of ensuring that the distributions of classification decisions when conditioned on a person's protected class status are the same for historically disadvantaged groups as they are for the majority.  The underlying implicit goal in many of these papers and associated discrimination measures is non-discrimination as we have defined it in this paper -- a decision-making process that is made based on an individual's attributes in the \acl{ts} and that does not have group skew in its mapping to the \acl{ds}.

The particular formulation of the \acl{gfm} goal has taken many forms.  Let $Pr[C = Yes | G = 0]$ be the probability of people in the minority group receiving a positive classification and $Pr[C= YES | G = 1]$ be the probability of people in the majority group receiving a positive classification.  Much previous work has considered the goal of achieving a low discrimination score \cite{Calders10NaiveBayes, Kamiran09Classifying, Kamishima11Fairness, Ruggieri2014TCloseness, icml2013_zemel13}, where the discrimination score is defined as $Pr[C = YES | G = 1] - Pr[C = YES | G = 0]$.  Since the goal is to bring this difference close to zero, the assumption is that groups should, as a whole, receive similar outcomes.  This reflects an underlying assumption of the \acl{wae} axiom so that similar group outcomes will be non-discriminatory.

Previous work \cite{Feldman2015DisparateImpact} has also created \acl{gfm}s with the goal of ensuring that decisions are non-discriminatory under the disparate impact four-fifths ratio, a U.S. legal notion with associated measure advocated by the E.E.O.C. \cite{eeoc1979}.  Work by Zafar et al. \cite{Zafar2015FairnessConstraints} has used a related definition that is easier to optimize.  The disparate impact four-fifths measure looks at the ratio comparing the protected class-conditioned probability of receiving a positive classification to the majority class' probability: $Pr[C=YES | G = 0] ~/~ Pr[C = 0 | G = 1]$.  Ratios that are closer to one are considered more fair, i.e., it is assumed that groups should as a whole receive similar classifications in order for the result to be non-discriminatory.  Again, this shows that the \acl{wae} axiom is being assumed in this measure.

Many of these works attempt to ensure non-discrimination by modifying the decision algorithm itself \cite{Calders10NaiveBayes, Kamishima11Fairness} while others change the outcomes after the decision has been drafted.  Especially interesting within the context of our definitional framework, some solutions change the input data to the machine learning algorithm before a model is trained \cite{Feldman2015DisparateImpact, icml2013_zemel13, Kamiran09Classifying}.  These works can be seen as attempting to reconstruct the \acl{ts} and make decisions directly based on that hypothesized reality under the \acl{wae} assumption.
